\documentclass[12pt]{article}
\usepackage{amsfonts}
\usepackage{amsmath}
\usepackage{amssymb}
\usepackage{amsthm}
\usepackage{geometry}
\usepackage{rotating}
\usepackage{multirow}
\usepackage[utf8]{inputenc}
\usepackage{array}
\usepackage[english]{babel}
\usepackage{float}
\usepackage{epstopdf}

\usepackage{bm}
\usepackage[numbers]{natbib}

\newcommand{\bfm}[1]   {\mbox{\boldmath{${#1}$}}}

\theoremstyle{plain} \newtheorem{lemma}{Lemma}
\theoremstyle{plain} 
\theoremstyle{plain} \newtheorem{definition}{Definition}
\theoremstyle{definition} \newtheorem*{algorithm}{Algorithm}

\begin{document}
\title{Coverage-adjusted confidence intervals for a binomial proportion}

\author
{ {M{\aa}ns Thulin}$^{1}$}
\date{}

\maketitle

\footnotetext[1]{Department of Mathematics, Uppsala University, Box 480, 751 06 Uppsala, Sweden.\\Phone: +46(0)184713389; E-mail: thulin@math.uu.se}

\begin{abstract}
\noindent We consider the classic problem of interval estimation of a proportion $p$ based on binomial sampling. The ''exact'' Clopper-Pearson confidence interval for $p$ is known to be unnecessarily conservative. We propose coverage-adjustments of the Clopper-Pearson interval using prior and posterior distributions of $p$. The adjusted intervals have improved coverage and are often shorter than competing intervals found in the literature. Using new heatmap-type plots for comparing confidence intervals, we find that the coverage-adjusted intervals are particularly suitable for $p$ close to 0 or 1.
\end{abstract}

\section{Introduction}\label{introduction}
Constructing a confidence interval for a proportion $p$ based on a binomial sample is a basic but important problem in statistics. Due to the discreteness of the binomial distribution, it is not possible to construct confidence intervals with exact coverage. Thus an interval based on normal approximation, known as the Wald interval, is taught in virtually every introductory statistics course. The interval is $\hat{p}\pm z_{\alpha/2}\sqrt{\hat{p}\hat{q}/n}$, where $\hat{p}=X/n$ is the sample proportion, $\hat{q}=1-\hat{p}$ and $z_{\alpha/2}$ is the $100(1-\alpha/2)$th percentile of the standard normal distribution.

Numerous authors have remarked on the surprisingly poor performance of the Wald interval. Errors in the approximation due to discreteness and skewness (for small $p$) can have significant impact on the coverage of the interval even for large $n$. In recent years, its weaknesses have been thoroughly investigated in comparisons of confidence intervals for $p$. \citet{bcd1,bcd2} gave examples of the erradic behaviour of the Wald interval, compared several intervals in terms of coverage and expected length and obtained general asymptotic results using Edgeworth expansions. \citet{pa1} compared twenty methods using different criteria. For recent developments and discussions, see for instance \citep{ck1,gm1,hu1,kp1,nn1,ne1}.

A natural alternative to the Wald interval is the Clopper-Pearson interval \citep{cp1}. It is based on the inversion of the equal-tailed binomial test and hence the interval contains all values of $p$ that aren't rejected by the test at confidence level $\alpha$. The lower limit is thus given by the value of $p_L$ such that
\begin{equation}\label{cp1}
\sum_{k=X}^n\binom{n}{k}p_L^k(1-p_L)^{n-k}=\alpha/2
\end{equation}
and the upper limit is given by the $p_U$ such that
\begin{equation}\label{cp2}
\sum_{k=0}^X\binom{n}{k}p_U^k(1-p_U)^{n-k}=\alpha/2.
\end{equation}
The computation of $p_L$ and $p_U$ is simplified by the following equality from \citet{kotz1}:
\[
\sum_{k=X}^n\binom{n}{k}p^k(1-p)^{n-k}=\int_0^pf(t,X,n-X+1)dt.
\]
where $f(t,r,s)$ is the density function of a $Beta(r,s)$ random variable. Consequently, the endpoints of the Clopper-Pearson interval $I_{CP}=(p_L,p_U)$ are beta quantiles:
\[
I_{CP}=\Big{(}B(\alpha/2,X,n-X+1),\quad B(1-\alpha/2,X+1,n-X)\Big{)}.
\]
$I_{CP}$ is exact in the sense that the minimum coverage over all $p$ is at least $1-\alpha$. For most values of $p$ however, especially values close to 0 or 1, it is far too conservative, giving a coverage that is much larger than the nominal coverage.

As several authors have pointed out \citep{ac1,bcd1,nn1} it is often more natural to study the \emph{mean} coverage rather than the minimum coverage. 
In this paper we construct coverage-adjusted Clopper-Pearson intervals with the mean coverage in mind, combining Bayesian and frequentist reasoning. The intervals are adjusted to have mean coverage $1-\alpha$ with respect to either a prior or a posterior distribution of $p$. The corrected intervals are seen to have several desirable properties in the frequentist setting.

A class of coverage-adjusted Clopper-Pearson intervals is introduced in Section \ref{coverage}. In Section \ref{comparisons} these intervals are compared to other popular intervals and new heatmap-style plots for comparing confidence intervals are introduced. The text concludes with a discussion in Section \ref{discussion} and an appendix with proofs, tables and several figures.


\section{Coverage-adjusted Clopper-Pearson intervals}\label{coverage}
\subsection{Definition}

As has already been mentioned, $I_{CP}$ is often unnecessarily conservative. This is illustrated in Figure \ref{cpcov}. It is clear from the figure that if we are willing to accept an interval which has a coverage less than $1-\alpha$ for \emph{some} values of $p$, the performance of $I_{CP}$ can be improved by choosing a larger $\alpha$, in which case the actual coverage would be closer to the desired coverage for \emph{most} values of $p$. The question, then, is how to choose the new $\alpha$. We propose that $\alpha'$ should be chosen to satisfy a mean coverage criterion.

\begin{figure}
\begin{center}
   \includegraphics[width=\textwidth]{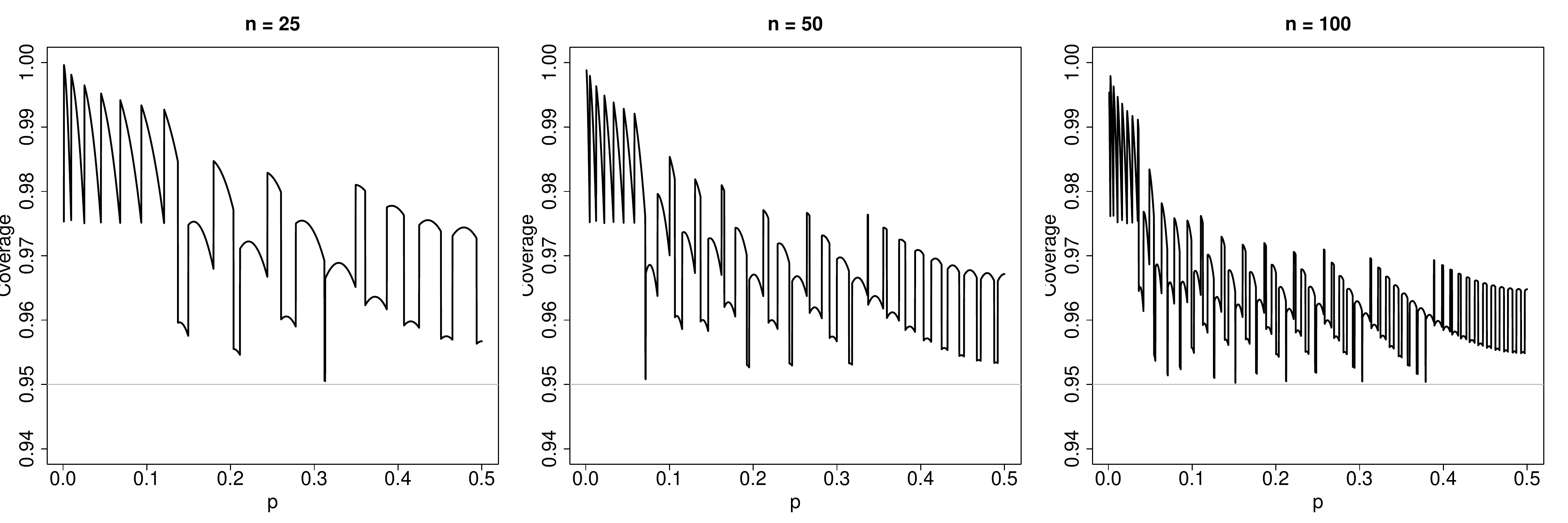}
   \caption{Actual coverage of the nominal 95 \% Clopper-Pearson interval.}\label{cpcov}
\end{center}
\end{figure}

\begin{definition}
Let $f(\cdot)$ be a density function on $(0,1)$. A mean coverage corrected $1-\alpha$ Clopper-Pearson interval $I_{GCP}=(p_L,p_U)$ is given by the unique solution to
\[\begin{split}
&\sum_{k=X}^n\binom{n}{k}p_L^k(1-p_L)^{n-k}=\alpha'/2,\\
&\sum_{k=0}^X\binom{n}{k}p_U^k(1-p_U)^{n-k}=\alpha'/2
\end{split}\]
where $\alpha'$ satisifies
\begin{equation}\label{condition}\begin{split}
C(\alpha',n)&=\int_0^1P(p\in I_{CP})\cdot f(p)dp\\
&=\int_0^1 \sum_{X=0}^n\bfm{1}(p\in I_{CP}(X,\alpha'))\binom{n}{X}p^X(1-p)^{n-X}\cdot f(p)dp=1-\alpha,
\end{split}\end{equation}
i.e. $\alpha'$ is such that the mean coverage of $I_{GCP}$ with respect to $f$ is $1-\alpha$.
\end{definition}
Note that this simply is the ordinary $1-\alpha'$ Clopper-Pearson interval, with $\alpha'$ chosen so that the mean coverage is $1-\alpha$. What differs is that $\alpha'$ needs to be determined before the endpoints are computed.

It should be pointed out that the adjusted intervals inherit important properties from $I_{CP}$. They are fully boundary-respecting, so that $I_{GCP}\subseteq (0,1)$, and equivariant in the sense of \cite{bs1}, meaning that the corresponding interval for $1-p$ is $(1-p_U,1-p_L)$. Furthermore, they have very favourable location properties in terms of the Box-Cox index of symmetry and balance of mesial and distal non-coverage, as described by \citet{ne1}. Finally, the minimum coverage over all $p$ is guaranteed to be at least $1-\alpha'$.

The choice of $f$ affects the performance of $I_{GCP}$ greatly. $f$ can be thought of as a weight function on $(0,1)$, used to put more weight on the performance for certain parts of the parameter space. In the following, we will refer to $f$ as being either a prior or posterior density, to show the connection between this weight function and Bayesian ideas.

\subsection{Prior mean coverage corrections}\label{prior}

The use of a prior distribution $f$ for coverage-adjustments can be motivated by the fact that in virtually all investigations, the experimenters will have some prior idea about how large $p$ is. In particular, it is often clear beforehand if $p$ is close to or far away from $1/2$.

$I_{CP}$ is symmetric in $p$ in the sense that the interval has the same properties for $p$ and $1-p$. For this reason, it is reasonable to use a symmetric prior for $p$. $Beta(r,r)$ priors, being conjugate priors of the binomial distribution, are a natural choice here. We divide the parameter space into three cases:

\emph{$p$ close to 0 or 1. } When $p$ is small, a prior with $r<1$ should be used, as such priors put more weight on the tails of the distribution. We will use the $Beta(1/2,1/2)$ prior in the following, but smaller $r$ can certainly be used. The coverage-adjustments will generally be larger for small $p$, as the overcoverage of $I_{CP}$ is largest in this part of the parameter space.

\emph{$p$ close to 1/4 or 3/4. } For medium-sized $p$, we wish to put approximately the same weight on the tails and the centre of the distribution. The uniform $Beta(1,1)$ prior is ideal for this. The resulting interval will however give a slight undercoverage for $p$ closer to 1/2, so if there is some worry that that $p$ may be above 0.40, say, a prior with $r$ slightly greater than 1 could be used. The interval constructed using the uniform prior seems to coincide with a corrected interval that was described informally by \citet{re1}.

\emph{$p$ close to 1/2. } If $p$ is believed to be closer to $1/2$, a prior with $r>1$ is recommendable. We will use the $Beta(2,2)$ prior. The coverage-adjustments will be smaller in this part of the parameter space, as $I_{CP}$ comes closest to attaining its nominal coverage around $p=1/2$.

\begin{figure}[H]
\begin{center}
   \includegraphics[width=\textwidth]{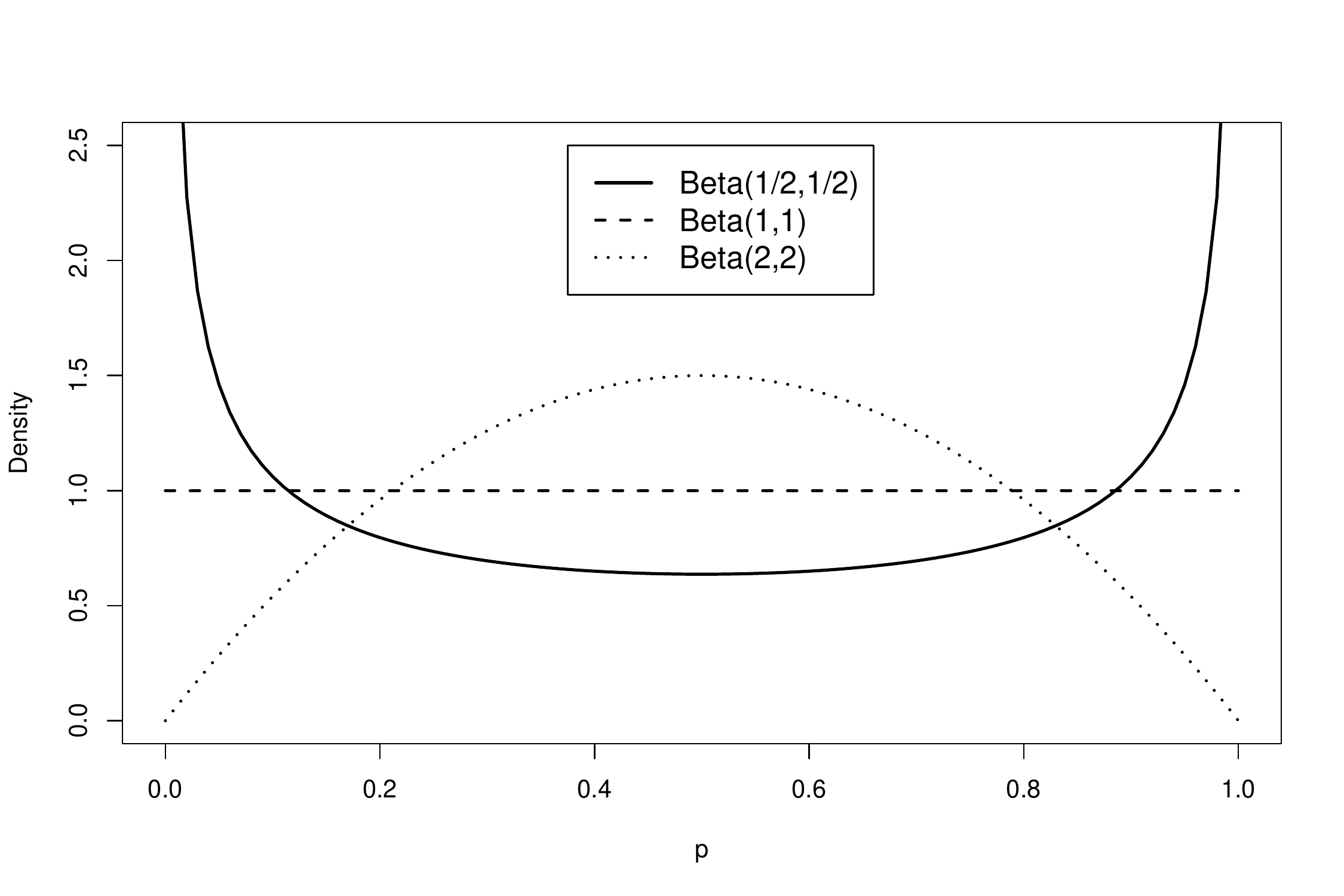}
   \caption{Three useful prior distributions for $p$.}\label{betapriors}
\end{center}
\end{figure}

\subsection{Posterior mean coverage corrections}\label{posterior}
Having used priors for coverage correction, it seems natural to consider using a posterior distribution of $p$ for coverage-adjustments, in order to get closer to the nominal coverage in areas of the parameters space that given the data are more likely to contain $p$.

With a $Beta(r,s)$ prior for $p$, the posterior distribution is $Beta(X+r,n-X+s)$ with density function
\[
f(p)=\frac{p^{X+r-1}(1-p)^{n-X+s-1}}{\beta(X+r,n-X+s)},\quad 0<p<1,
\]
where $\beta(\cdot,\cdot)$ is the beta function. Thus the posterior coverage corrected Clopper-Pearson interval $I_{GCP}$ is, given $X$, determined by the condition (\ref{condition}) with the function
\[
C(\alpha',n,X)=\int_0^1 \sum_{Y=0}^n\bfm{1}(p\in I_{CP}(Y,\alpha'))\binom{n}{Y}p^Y(1-p)^{n-Y}\cdot \frac{p^{X+r-1}(1-p)^{n-X+s-1}}{\beta(X+r,n-X+s)}dp.
\]
In the comparison later in the text, we will use the $Beta(1/2,1/2)$, $Beta(1,1)$ and $Beta(2,2)$ priors, with the same reasoning as in the previous section.

Conditioning the coverage-adjustments on the data may seem hazardous in a frequentist setting, but as we will demonstrate in Section \ref{comparisons}, this approach leads to short confidence intervals with good coverage properties.


\subsection{Determining the adjusted confidence level}
While $\alpha'$ can be approximated by using an asymptotic expansion for the coverage to solve the equation (\ref{condition}) approximately, it is more convenient to use a numerical method with exact coverages. The following lemma ensures that $C$ is continuous and decreasing in $\alpha'$. This guarantees that $\alpha'$ easily can be found numerically by using for instance bisection to solve the equation $C(\alpha',n,r,s)=1-\alpha$. The proof of the lemma is given in the Appendix, along with a table of $\alpha'$ for different choices of $\alpha$ and $n$ for $f(p)=1$.
\begin{lemma}\label{app1}
Let $I_{CP}(X,\alpha)=(p_L(X,\alpha),p_U(X,\alpha))$ be the $1-\alpha$ Clopper-Pearson interval and let $f(p,r,s)$, $0<p<1$, be the density of the $Beta(r,s)$ distribution. The mean coverage of $I_{CP}(X,\alpha)$ with respect to the density $f(p)$,
\[
C(\alpha,n,r,s)=\int_0^1 \sum_{X=0}^n\bfm{1}(p\in I_{CP}(X,\alpha))\binom{n}{X}p^X(1-p)^{n-X}f(p,r,s)dp,
\]
is continuous and strictly decreasing in $\alpha$.
\end{lemma}
The algorithm for finding $\alpha'$ using bisection is as follows.
\begin{algorithm}
Given a tolerance $tol$, $\alpha$, $n$ and a density $f$:
\begin{enumerate}
\item Start with an initial lower bound $\alpha'_{L,0}=\alpha$ and an upper bound $\alpha'_{U,0}$. The initial guess is $\alpha'_0=(\alpha'_{L,0}+\alpha'_{U,0})/2$.
\item Set $i=0$.
\item While $|1-\alpha-C(\alpha'_i,n)|>tol$:
\begin{itemize}
\item If $C(\alpha'_i,n)>1-\alpha$ then $\alpha'_{L,i+1}=\alpha'_i$, $\alpha'_{U,i+1}=\alpha'_{U,i}$ and $\alpha'_{i+1}=(\alpha'_{L,i+1}+\alpha'_{U,i+1})/2$.
\item Else $\alpha'_{U,i+1}=\alpha'_i$, $\alpha'_{L,i+1}=\alpha'_{L,i}$ and $\alpha'_{i+1}=(\alpha'_{L,i+1}+\alpha'_{U,i+1})/2$.
\item $i=i+1$.
\end{itemize}
\item $\alpha'=\alpha'_i$.
\end{enumerate}
\end{algorithm}
For the algorithm to converge, two conditions must be satisfied. First, $\alpha'\leq\alpha'_{U,0}$, i.e. $\alpha'$ must not exceed the upper bound. Second, $C(\alpha'_i,n)$ must be computed with sufficient precision ($tol$ determines what is sufficient).

Implementations of the above algorithm in R and MS Excel are available from the author.


\subsection{An example}
We illustrate the use of the coverage-adjustments with clinical data from an influenza vaccine study performed by \citet{he1}. $n=96$ fully vaccinated children younger than 2 years were included in the study. $X=4$ of these contracted influenza during the 2007-08 influenza season.

The $95 \%$ Clopper-Pearson interval for the proportion of vaccinated children younger than 2 years that will contract influenza is $( 0.012,0.103)$. Using a prior $Beta(1,1)$ correction, we get $\alpha'\approx 0.06967$. Letting $B(\cdot,r,s)$ be the quantile function of the $Beta(r,s)$ distribution, the coverage-adjusted confidence interval is 
\[
\Big{(}B(0.06967/2,4,93),~ B(1-0.06967/2,5,92)\Big{)}=( 0.013,0.098).
\]
Using a posterior $Beta(1/2,1/2)$ correction, $\alpha'\approx 0.09385$ and the interval is $( 0.014,0.094)$.

\section{Comparison of intervals}\label{comparisons}

\subsection{Other intervals}\label{othint}
Following the comparison performed by \citet{bcd1}, two confidence intervals for $p$ have emerged as being the intervals to which all other intervals should be compared. These are the Wilson and Jeffreys prior intervals, presented next.

\emph{The Wilson interval. } Like the Wald interval, the \citet{wi1} score interval is based on an inversion of the large sample normal test
\[
\Big{|}\frac{\hat{p}-p}{d(\hat{p})}\Big{|}\leq z_{\alpha/2},
\]
where $d(\hat{p})$ is the standard error of $\hat{p}$. Unlike the Wald interval, however, the inversion is obtained using the null standard error $(pq/n)^{1/2}$ instead of the sample standard error $(\hat{p}\hat{q}/n)^{1/2}$. The solution of the resulting quadratic equation leads to the confidence interval
\[
I_W=\frac{X+z_{\alpha/2}^2/2}{n+z_{\alpha/2}^2}\pm \frac{z_{\alpha/2}}{n+z_{\alpha/2}^2}\sqrt{\hat{p}\hat{q}n+z_{\alpha/2}^2/4}.
\]
$I_W$ typically has coverage close to the nominal coverage and comparatively short expected length. Indeed, it can be shown \citep{bcd2} that $I_W$ has some near-optimal length properties among intervals with nominal coverage $1-\alpha$. $I_W$ is therefore the natural benchmark for new confidence intervals.

The main drawback of $I_W$ is that its coverages oscillates too much for fixed $n$ and $p$ close to 0 or 1. Recently, \citet{yg1} proposed a small modification of the interval that solves this problem. The improved coverage comes at the cost of a slightly wider interval. As the results of our comparison would not change qualitatively if the Guan interval were to be used instead of $I_W$, we stick to the more familiar unmodified version.

\emph{The Jeffreys prior interval. } Let $X\sim Bin(n,p)$ and let $p$ have prior distribution $Beta(r,s)$. Then the posterior distribution is $Beta(X+r,n-X+s)$ and letting $B(\alpha,r,s)$ denote the $\alpha$-quantile of the $Beta(r,s)$ distribution, a $100(1-\alpha) \%$ Bayesian interval is
\[
\Big{(}B(\alpha/2,X+r,n-X+s),\quad  B(1-\alpha/2,X+r,n-X+s)\Big{)}.
\]
\citet{pa1} used the uniform prior $r=s=1$ in their comparison, whereas \citet{bcd1} used the Jeffreys prior $r=s=1/2$. The difference between the two intervals is small. We use the latter and denote it $I_J$.

$I_J$ has performance close to that of the $I_W$, and is often prefered when $p$ is believed to be close to 0 or 1.

\subsection{Numerical comparison}\label{exactcomp}
In Figure \ref{figCov1} the actual coverages of some confidence intervals with nominal coverage $95 \%$ are shown for $p\in( 0,0.5\rbrack$ when $n=25$. All intervals are equivariant in the sense that the coverage is the same for $p$ and $1-p$. Note that the coverages have been computed exactly (up to machine epsilon) and thus haven't been obtained by simulation. 

The Wilson interval has fairly good coverage properties when $p$ isn't close to 0, in which case it oscillates wildly. The Jeffreys prior interval has similar coverage, except for one big dip for a moderately sized $p$. The coverage-adjusted intervals tend to have good coverage properties in the areas dictated by the prior distribution used for the adjustment. Some of the prior corrected intervals suffer from either undercoverage or overcoverage for the parts of the parameter space that $f$ put low weight on.

The expected length of the intervals are shown in Figure \ref{figLen1}. In many cases, the corrected intervals have shorter expected length than the Wilson and Jeffreys prior intervals. For some intervals, this is due to undercoverage caused by $f$ putting more weight on a different part of the parameter space, but in some cases it is a consequence of a succesful correction.

In order to compare intervals over the entire parameter space for different values of $n$, we use heatmap-type plots in Figures \ref{figExp1}-\ref{figExp6}. Studying the plots for coverage and expected length at the same time gives us a good way of comparing the intervals. The heatmap-type plots, combined with more traditional plots such as those in Figures \ref{figCov1}-\ref{figLen1}, give a more complete comparison of different intervals than what has previously been possible.

The Wilson interval is compared to the prior corrected Clopper-Pearson $Beta(1,1)$ interval in Figures \ref{figExp1}-\ref{figExp2} and to the posterior corrected Clopper-Pearson $Beta(0.5,0.5)$ interval in Figures \ref{figExp3}-\ref{figExp4}. The corrected intervals simultaneously offer greater coverage and shorter expected length for small $p$. The difference is larger for small $n$ and is particularly noticeable at the 99 \% confidence level.

In the comparison between the Jeffreys prior interval and the posterior corrected Clopper-Pearson $Beta(0.5,0.5)$ interval in Figures \ref{figExp5}-\ref{figExp6}, the corrected interval is found to offer at least as short intervals with the same actual coverage as the Jeffreys prior interval.


\section{Discussion}\label{discussion}
\subsection{Conclusions}
We introduced coverage-adjusted Clopper-Pearson intervals, where the intervals are adjusted to give mean coverage $1-\alpha$ with respect to either a prior or posterior distribution of $p$. We investigated the properties of several such intervals. The numerical results were presented graphically, partially with new heatmap-type plots.

In the comparison with the benchmark Wilson and Jeffreys prior intervals, we found the coverage-adjusted Clopper-Pearson intervals to be preferable if $p$ is believed to be close to 0 or 1, as these intervals have both better coverage and shorter expected length in this setting. We have thus seen that it is possible to improve upon the Wilson and Jeffreys prior intervals for $p$ close to 0 or 1, if we are willing to accept that we use intervals that may have bad coverage properties in regions of the parameter space that are far from where our prior information indicates that $p$ is.

In conclusion, the coverage-adjusted Clopper-Pearson intervals seem to be strong competitors against other methods for constructing confidence intervals for small binomial proportions. For $p$ closer to 0.5, the Wilson interval seems to be preferable.

\subsection{Further developments}
The extension of the ideas presented here to one-sided intervals and to other distributions, such as the Poisson and negative binomial distributions is straightforward. Likewise, it should be possible to apply such corrections to tests about the difference of two binomial proportions. It remains to be seen whether the corrections yield intervals with interesting properties in these cases as well.

Apart from mean coverage, several other conditions can be used to ensure that the coverage is close to $1-\alpha$ on average. Examples include median coverage conditions, minimum mean squared coverage error conditions and minimum absolute coverage error conditions.

\subsubsection*{Acknowledgements}
The author wishes to thank an anonymous reviewer and Silvelyn Zwanzig for several helpful comments and Sven Erick Alm, who proposed the idea of a posterior correction. All figures were produced using R.

\pagebreak
\appendix
\section{Appendix}

\subsection{Continuity and monotonicity of the mean coverage}
\begin{proof}[Proof of Lemma \ref{app1}]
Changing the order of summation of integration, the mean coverage can be rewritten as
\[\begin{split}
C(\alpha,n,r,s)&=\sum_{X=0}^n\binom{n}{X}\int_0^1\bfm{1}(p\in I_{CP}(X,\alpha))p^X(1-p)^{n-X}f(p,r,s)dp\\
&= \sum_{X=0}^n\binom{n}{X}\int_{p_L(X,\alpha)}^{p_U(X,\alpha)}p^X(1-p)^{n-X}f(p,r,s)dp.
\end{split}\]
Since sums of continuous functions are continuous, it suffices to show that
\begin{equation}\label{cont1}
\int_{p_L(X,\alpha)}^{p_U(X,\alpha)}p^X(1-p)^{n-X}f(p,r,s)dp
\end{equation}
is continuous for fixed $X$, $n$, $r$ and $s$. Seeing as $f(p)=p^{r-1}(1-p)^{s-1}/\beta(r,s)$, this definite integral is a polynomial in $p_L(X,\alpha)$ and $p_U(X,\alpha)$. Since the quantile functions of the beta distributions, and thus the limits of integration, are continuous in $\alpha$, the continuity of $C(\alpha,n,r,s)$ in $\alpha$ follows.

Similarly, as $p_L(X,\alpha)$ is strictly increasing in $\alpha$ and $p_U(X,\alpha)$ is strictly decreasing in $\alpha$, and since $p^X(1-p)^{n-X}f(p,r,s)\geq0$ for all $p$, the definite integral (\ref{cont1}) is strictly decreasing in $\alpha$. $C(\alpha,n,r,s)$ is the sum of $n+1$ strictly decreasing functions and thus also strictly decreasing.
\end{proof}

\subsection{Tables}
For a given prior distribution, $\alpha'$ is easily computed numerically given $n$, $\alpha$ and, in the case of a posterior correction, $X$. We give a table of $\alpha'$ for the prior corrected Clopper-Pearson $Beta(1,1)$ interval as an example.
\begin{table}[H]
\begin{center}
\caption{$\alpha'$ for the prior corrected Clopper-Pearson $Beta(1,1)$ interval.}\label{alfatabell}\
\begin{tabular}{|c| c c | c| c c | c| c c | }
\hline
$n$ &  $\alpha=0.05$ & $\alpha=0.01$ & $n$ &  $\alpha=0.05$ & $\alpha=0.01$ & $n$ &  $\alpha=0.05$ & $\alpha=0.01$ \\ \hline
5 & 0.1772 & 0.0516& 55&0.0769 &0.0171 & 110& 0.0682&0.0147 \\
10 &0.1280 & 0.0331& 60&0.0756 &0.0167 & 120&0.0674 &0.0145 \\
15 &0.1095 & 0.0269& 65&0.0745 &0.0164 & 130&0.0666  &0.0143 \\
20 &0.0995 & 0.0237& 70&0.0735 &0.0161 & 140&0.0660 &0.0141\\
25 &0.0931 & 0.0218& 75&0.0726 &0.0159 & 150&0.0654 &0.0139 \\
30 &0.0885 & 0.0204& 80&0.0718 &0.0156 & 160& 0.0649& 0.0138 \\
35 &0.0851 & 0.0194& 85&0.0710 &0.0154 & 170& 0.0644&0.0137 \\
40 &0.0825 & 0.0186& 90&0.0704 &0.0153 & 180&0.0640 & 0.0135\\
45 &0.0803 & 0.0180& 95&0.0698 &0.0151 & 190&0.0636 &0.0134 \\
50 &0.0785 & 0.0175& 100&0.0692 &0.0150 & 200&0.0632 & 0.0133\\
\hline
\end{tabular}\\
\end{center}
\end{table}

Next, we give some examples of the uncorrected, prior corrected $Beta(1,1)$ and posterior corrected $Beta(1/2,1/2)$ Clopper-Pearson intervals for a few combinations of $n$, $X$ and $\alpha$. The posterior corrected $Beta(1/2,1/2)$ tends to get a larger correction, and thus shorter intervals, than the prior corrected $Beta(1,1)$ interval if the observed $X$ is close to $0$ or $n$ and a smaller correction if $X$ is close to $n/2$.

\begin{table}[H]
\begin{center}
\caption{Some examples of coverage-adjusted Clopper-Pearson intervals.}\label{alfatabell2}\
\begin{tabular}{c|c|cc|cc|cc|}
\cline{3-8}
\multicolumn{2}{c|}{}&\multicolumn{2}{c|}{No correction}&\multicolumn{2}{c|}{Prior $Beta(1,1)$}&\multicolumn{2}{c|}{Posterior $Beta(\frac{1}{2},\frac{1}{2})$} \\
\hline
\multicolumn{1}{|c|}{$n$} & $X$  & $\alpha=0.05$ & $\alpha=0.01$ & $\alpha=0.05$ & $\alpha=0.01$ & $\alpha=0.05$ & $\alpha=0.01$\\ \hline
\multicolumn{1}{|c|}{$20$} & $1$ & $(0.0012,$ & $(0.0003,$& $(0.0025,$ & $(0.0006,$ & $(0.0042,$ & $(0.0010,$ \\
\multicolumn{1}{|c|}{}&  & $~0.2487)$ & $~0.3171)$& $~0.2163)$ & $~0.2815)$ & $~0.1925)$ & $~0.2591)$ \\[1.25mm]
\multicolumn{1}{|c|}{}& $2$ & $(0.0123,$ & $(0.0053,$& $(0.0180,$ & $(0.0083,$ & $(0.0207,$ & $(0.0099,$ \\
\multicolumn{1}{|c|}{}&  & $~0.3170)$ & $~0.3871)$& $~0.2829)$ & $~0.3509)$ & $~0.2697)$ & $~0.3364)$ \\[1.25mm]
\multicolumn{1}{|c|}{}& $5$ & $(0.0866,$ & $(0.0583,$& $(0.1039,$ & $(0.0718,$ & $(0.1013,$ & $(0.0705,$ \\
\multicolumn{1}{|c|}{}&  & $~0.4910)$ & $~0.5598)$& $~0.4559)$ & $~0.5248)$ & $~0.4608)$ & $~0.5281)$ \\[1.25mm]
\multicolumn{1}{|c|}{}& $10$ & $(0.2720,$ & $(0.2177,$& $(0.3017,$ & $(0.2447,$ & $(0.2929,$ & $(0.2369,$ \\
\multicolumn{1}{|c|}{}&  & $~0.7280)$ & $~0.7823)$& $~0.6983)$ & $~0.7553)$ & $~0.7071)$ & $~0.7631)$ \\[1.25mm]
  \hline
  \multicolumn{1}{|c|}{$50$} & $1$ & $(0.0005,$ & $(0.0001,$& $(0.0008,$ & $(0.0002,$ & $(0.0016,$ & $(0.0004,$ \\
  \multicolumn{1}{|c|}{}&  & $~0.1065)$ & $~0.1394)$& $~0.0967)$ & $~0.1282)$ & $~0.0812)$ & $~0.1122)$ \\[1.25mm]
\multicolumn{1}{|c|}{}& $5$ & $(0.0333,$ & $(0.022,$& $(0.0376,$ & $(0.0255,$ & $(0.0387,$ & $(0.0267,$ \\
\multicolumn{1}{|c|}{}&  & $~0.2181)$ & $~0.2580)$& $~0.2058)$ & $~0.2448)$ & $~0.2027)$ & $~0.2402)$ \\[1.25mm]
\multicolumn{1}{|c|}{}& $12$ & $(0.1306,$ & $(0.1056,$& $(0.1395,$ & $(0.1133,$ & $(0.1378,$ & $(0.1120,$ \\
\multicolumn{1}{|c|}{}&  & $~0.3817)$ & $~0.4255)$& $~0.3676)$ & $~0.4112)$ & $~0.3702)$ & $~0.4136)$ \\[1.25mm]
\multicolumn{1}{|c|}{}& $25$ & $(0.3553,$ & $(0.3155,$& $(0.3686,$ & $(0.3282,$ & $(0.3644,$ & $(0.3242,$ \\
\multicolumn{1}{|c|}{}&  & $~0.6447)$ & $~0.6845)$& $~0.6314)$ & $~0.6718)$ & $~0.6356)$ & $~0.6757)$ \\[1.25mm]
  \hline
   \multicolumn{1}{|c|}{$100$} & $1$ & $(0.0003,$ & $(0.0001,$& $(0.0004,$ & $(0.0001,$ & $(0.0008,$ & $(0.0002,$ \\
\multicolumn{1}{|c|}{}&  & $~0.0545)$ & $~0.0720)$& $~0.0508)$ & $~0.0677)$ & $~0.0413)$ & $~0.0577)$ \\[1.25mm]
\multicolumn{1}{|c|}{}& $10$ & $(0.0490,$ & $(0.0382,$& $(0.0518,$ & $(0.0405,$ & $(0.0523,$ & $(0.0410,$ \\
\multicolumn{1}{|c|}{}&  & $~0.1762)$ & $~0.2020)$& $~0.1705)$ & $~0.1959)$ & $~0.1695)$ & $~0.1946)$ \\[1.25mm]
\multicolumn{1}{|c|}{}& $25$ & $(0.1688,$ & $(0.1477,$& $(0.1739,$ & $(0.1525,$ & $(0.1728,$ & $(0.1515,$ \\
\multicolumn{1}{|c|}{}&  & $~0.3466)$ & $~0.3769)$& $~0.3396)$ & $~0.3698)$ & $~0.3410)$ & $~0.3712)$ \\[1.25mm]
\multicolumn{1}{|c|}{}& $50$ & $(0.3983,$ & $(0.3689,$& $(0.4052,$ & $(0.3756,$ & $(0.4031,$ & $(0.3735,$ \\
\multicolumn{1}{|c|}{}&  & $~0.6017)$ & $~0.6311)$& $~0.4948)$ & $~0.6244)$ & $~0.5969)$ & $~0.6265)$ \\[1.25mm]
  \hline

\end{tabular}
\end{center}
\end{table}

\pagebreak

\subsection{Figures}

\pagestyle{empty}

\begin{figure}[H]
\begin{center}
\includegraphics[width=\textwidth]{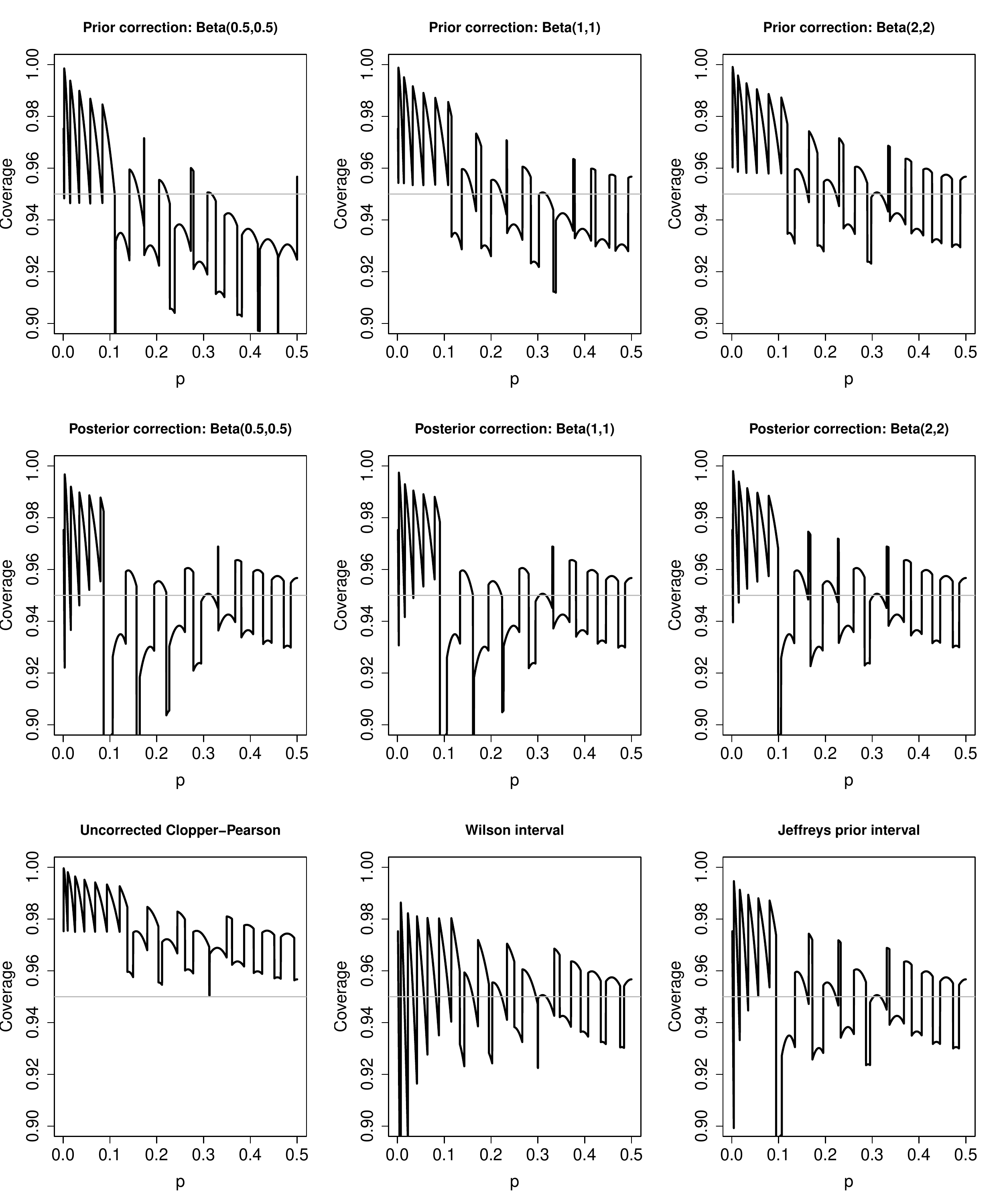}
   \caption{Coverage of nominal 95 \% intervals for $n=25$.}\label{figCov1}
\end{center}
\end{figure}

\begin{figure}[H]
\begin{center}
\includegraphics[width=\textwidth]{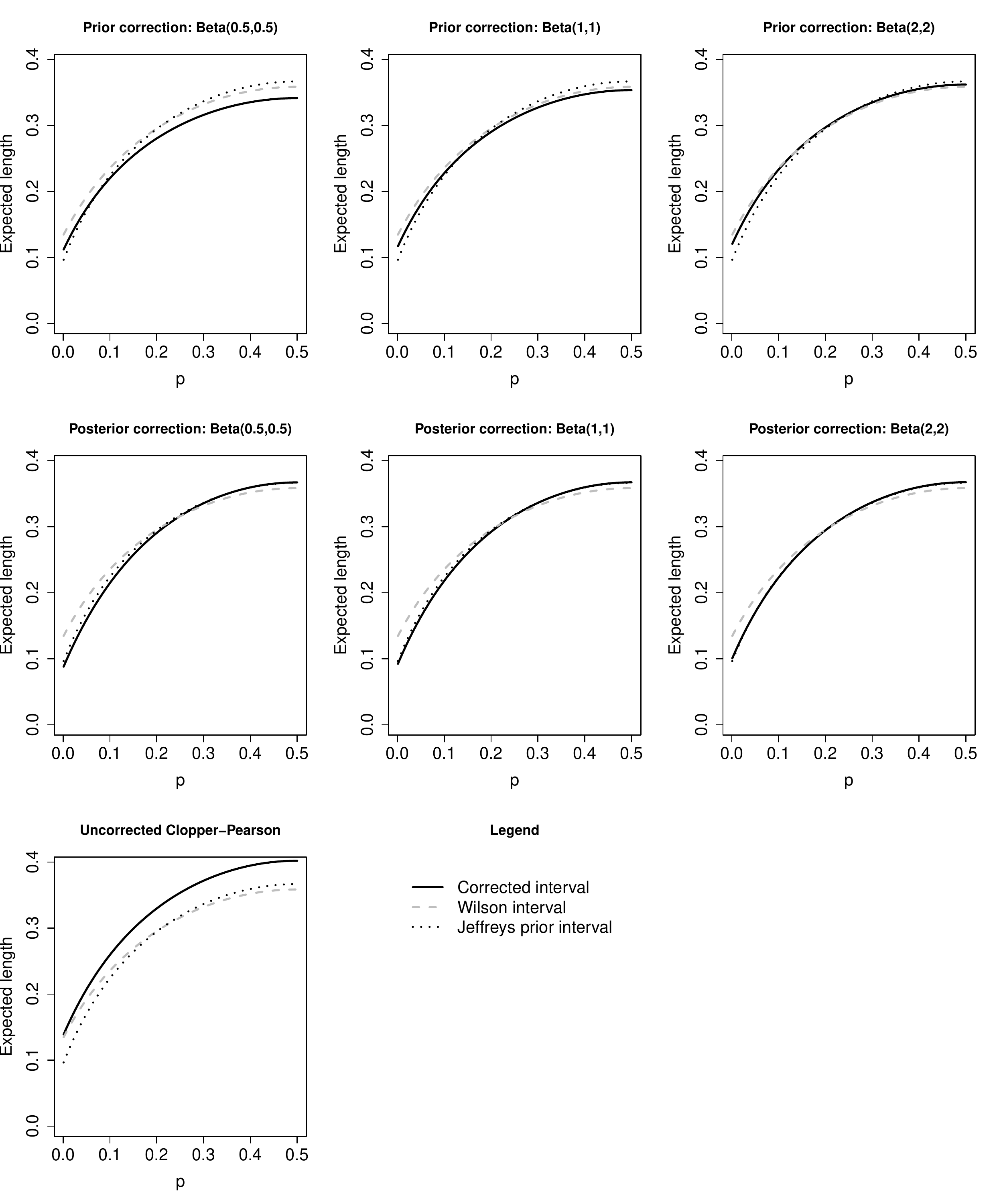}
   \caption{Expected length of nominal 95 \% intervals for $n=25$.}\label{figLen1}
\end{center}
\end{figure}

%
%
%
%
%
%
%

%

\begin{figure}[H]
\begin{center}
\includegraphics[width=\textwidth]{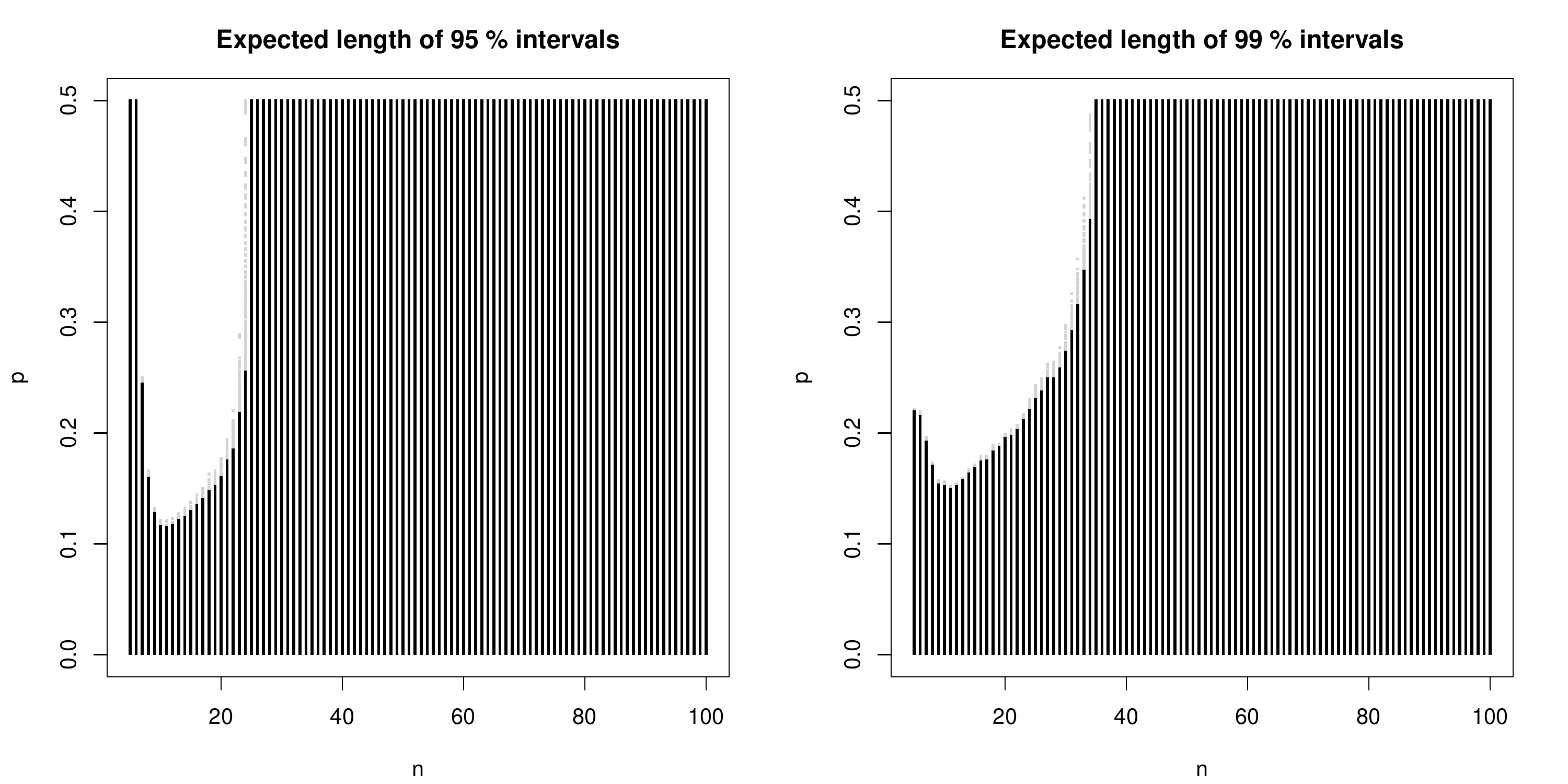}
   \caption{Comparison of the expected length of the Wilson and prior corrected Clopper-Pearson $Beta(1,1)$ intervals for different $n$ and $p$. In the black points, the prior corrected interval has shorter expected length. In the grey points, the intervals have equal expected lengths. Based on a grid of 500 equidistant values of $p$ in $( 0.001,0.5)$.}\label{figExp1}
   \vspace{3mm}
\includegraphics[width=\textwidth]{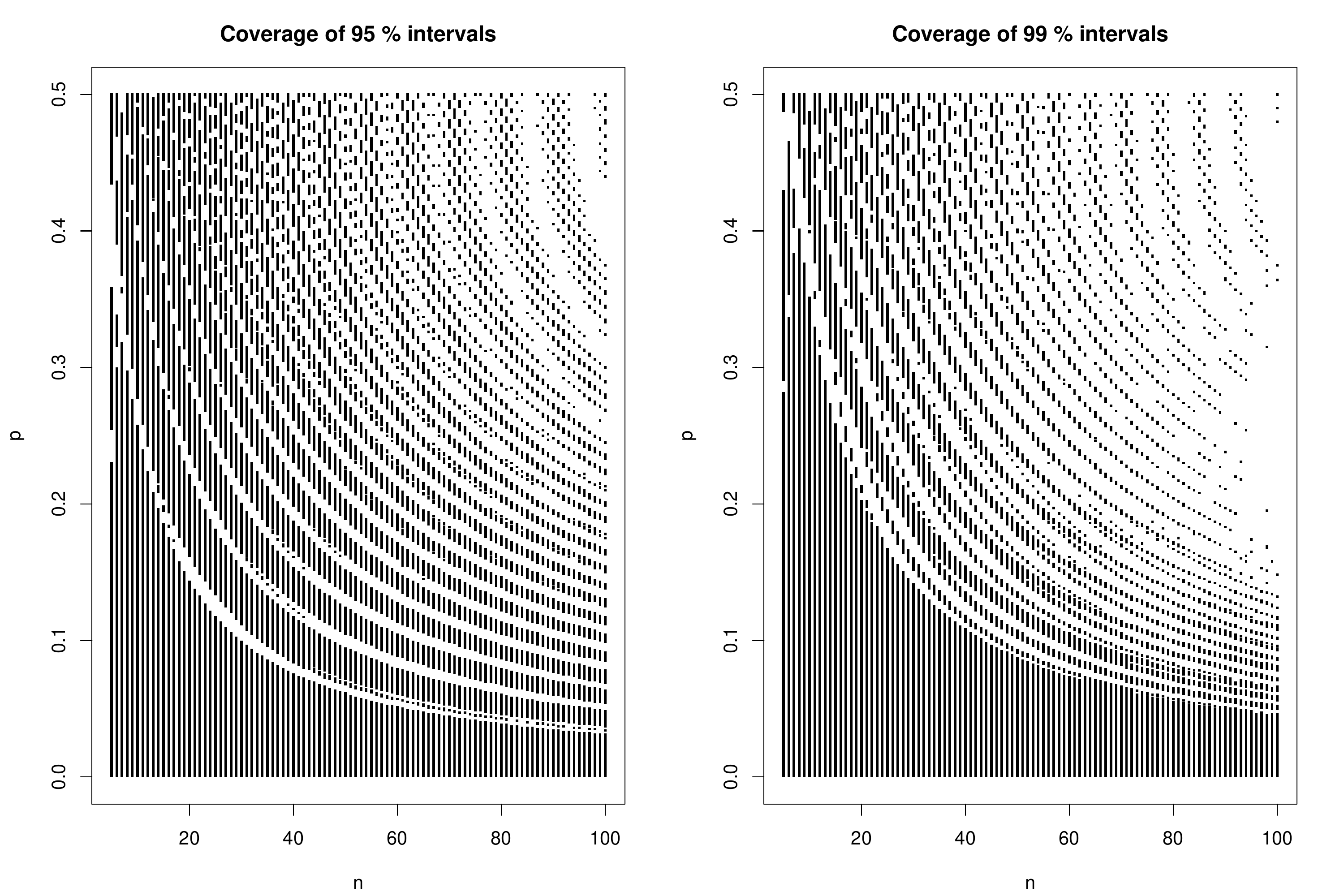}
   \caption{Comparison of the coverage of the Wilson and prior corrected Clopper-Pearson $Beta(1,1)$ intervals for different $n$ and $p$. In the black points the prior corrected interval has greater coverage, in the white points the Wilson interval has greater coverage and in the grey points the intervals have equal coverage (when rounded to 3 decimal places).}\label{figExp2}
\end{center}
\end{figure}

\begin{figure}[H]
\begin{center}
\includegraphics[width=\textwidth]{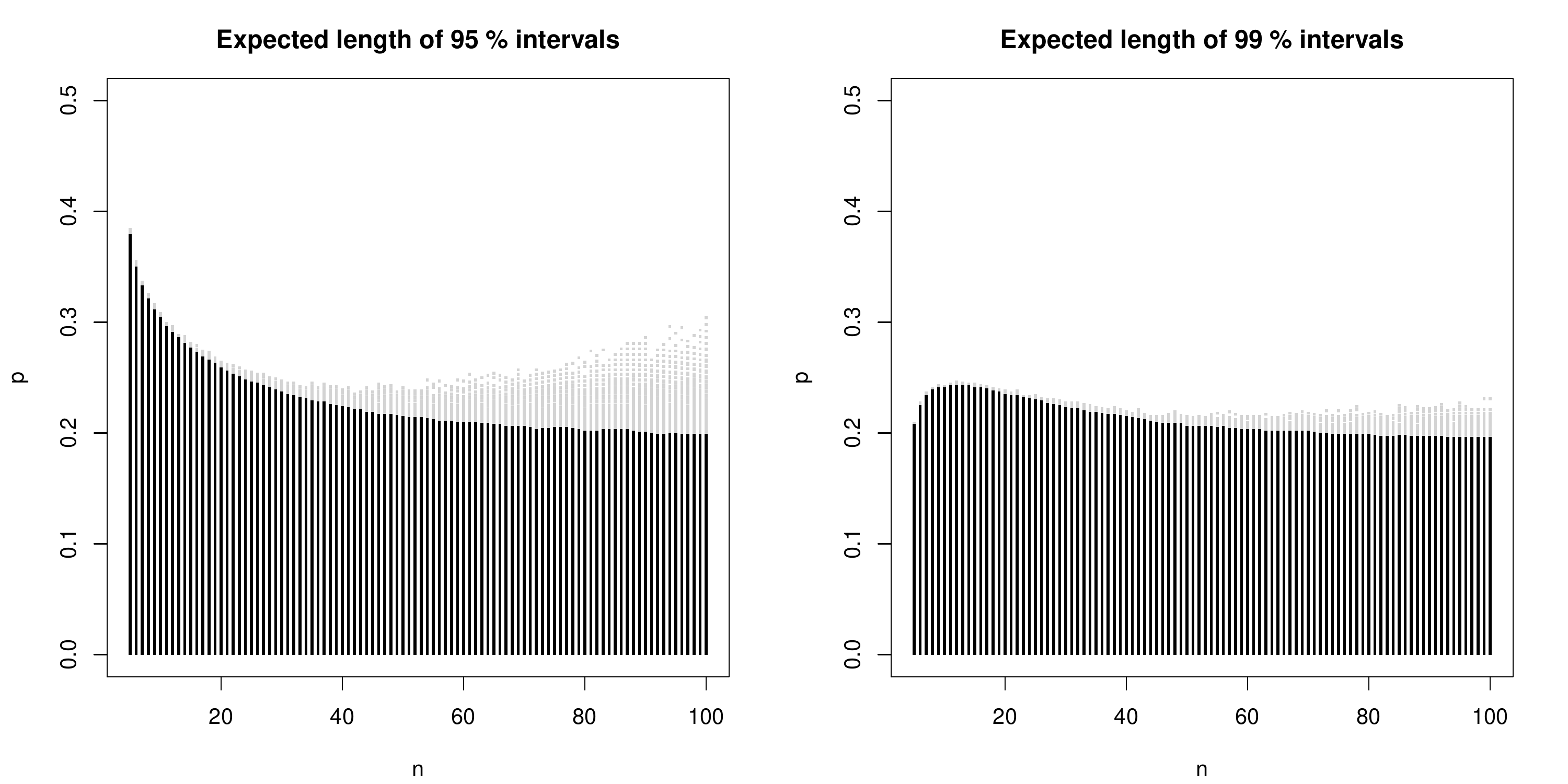}
   \caption{Comparison of the expected length of the Wilson and posterior corrected Clopper-Pearson $Beta(0.5,0.5)$ interval intervals for different $n$ and $p$. In the black points, the posterior corrected interval has shorter expected length. In the grey points, the intervals have equal expected lengths.}\label{figExp3}
      \vspace{3mm}
      \includegraphics[width=\textwidth]{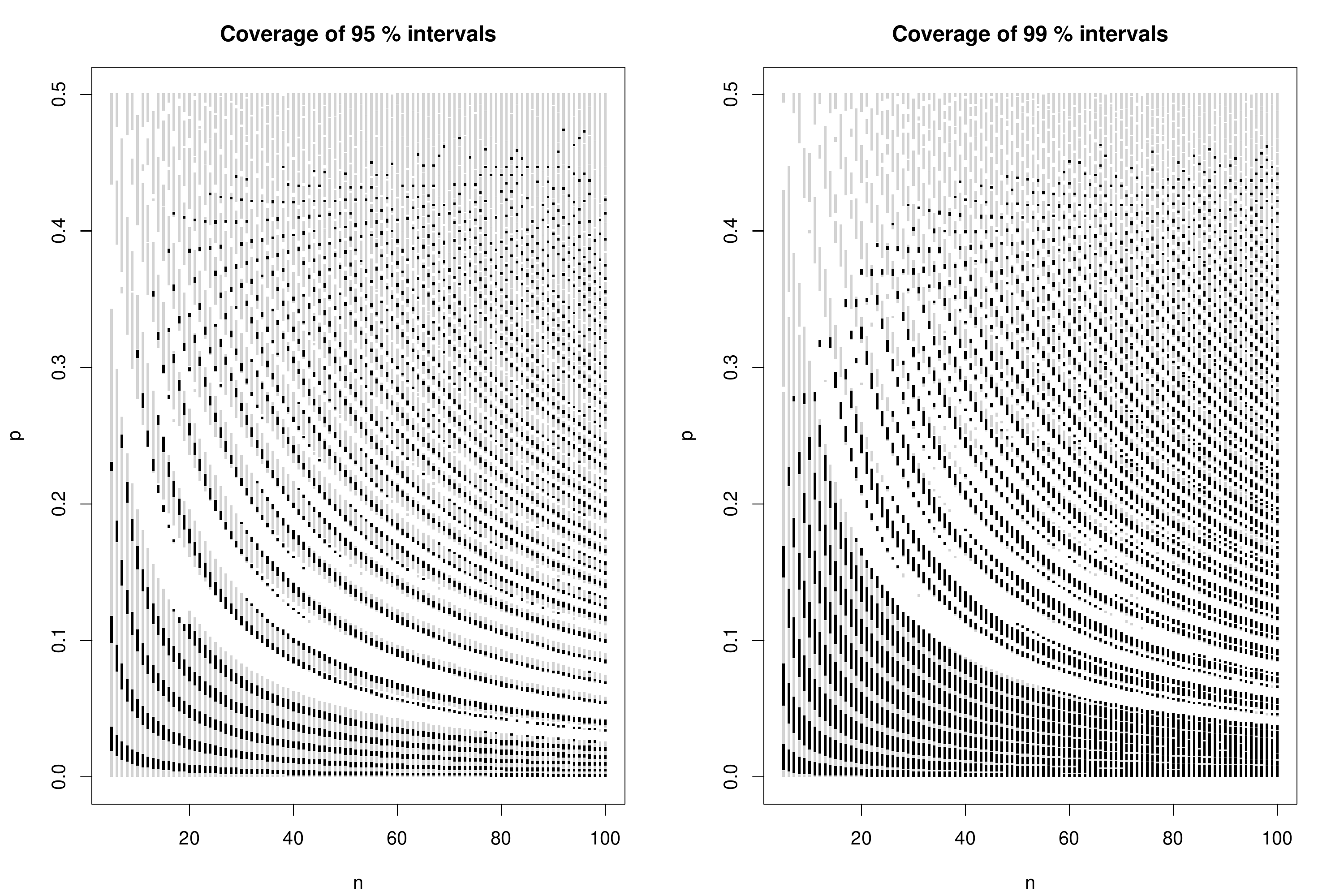}
   \caption{Comparison of the coverage of the Wilson and posterior corrected Clopper-Pearson $Beta(0.5,0.5)$ intervals for different $n$ and $p$. In the black points the posterior corrected interval has greater coverage, in the white points the Wilson interval has greater coverage and in the grey points the intervals have equal coverage (when rounded to 3 decimal places).}\label{figExp4}

\end{center}
\end{figure}

\begin{figure}[H]
\begin{center}
   \includegraphics[width=\textwidth]{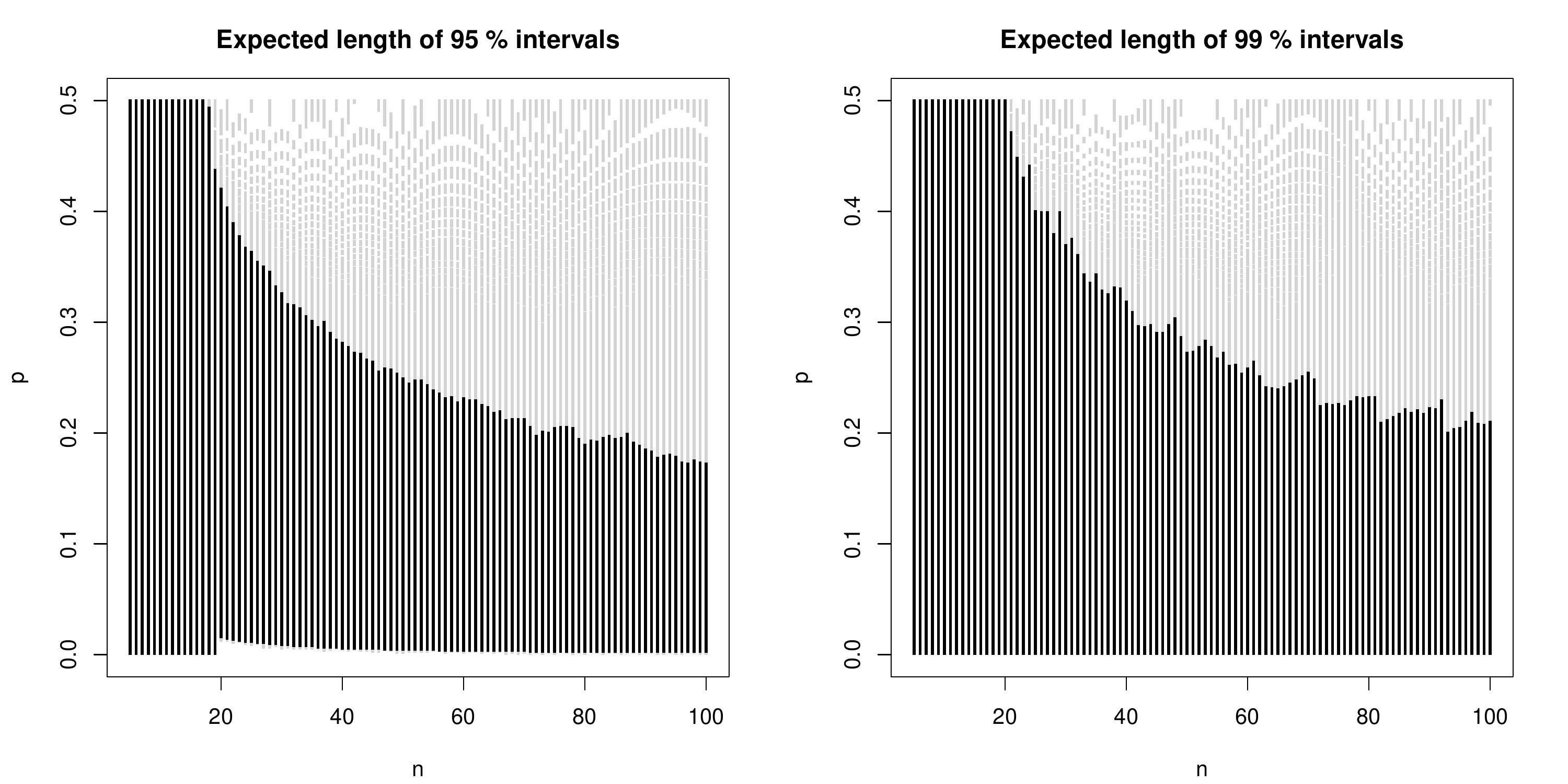}
   \caption{Comparison of the expected length of the Bayesian Jeffreys prior and posterior corrected Clopper-Pearson $Beta(0.5,0.5)$ intervals for different $n$ and $p$. In the black points, the posterior corrected interval has shorter expected length. In the grey points, the intervals have equal expected lengths.}\label{figExp5}
      \vspace{3mm}
   \includegraphics[width=\textwidth]{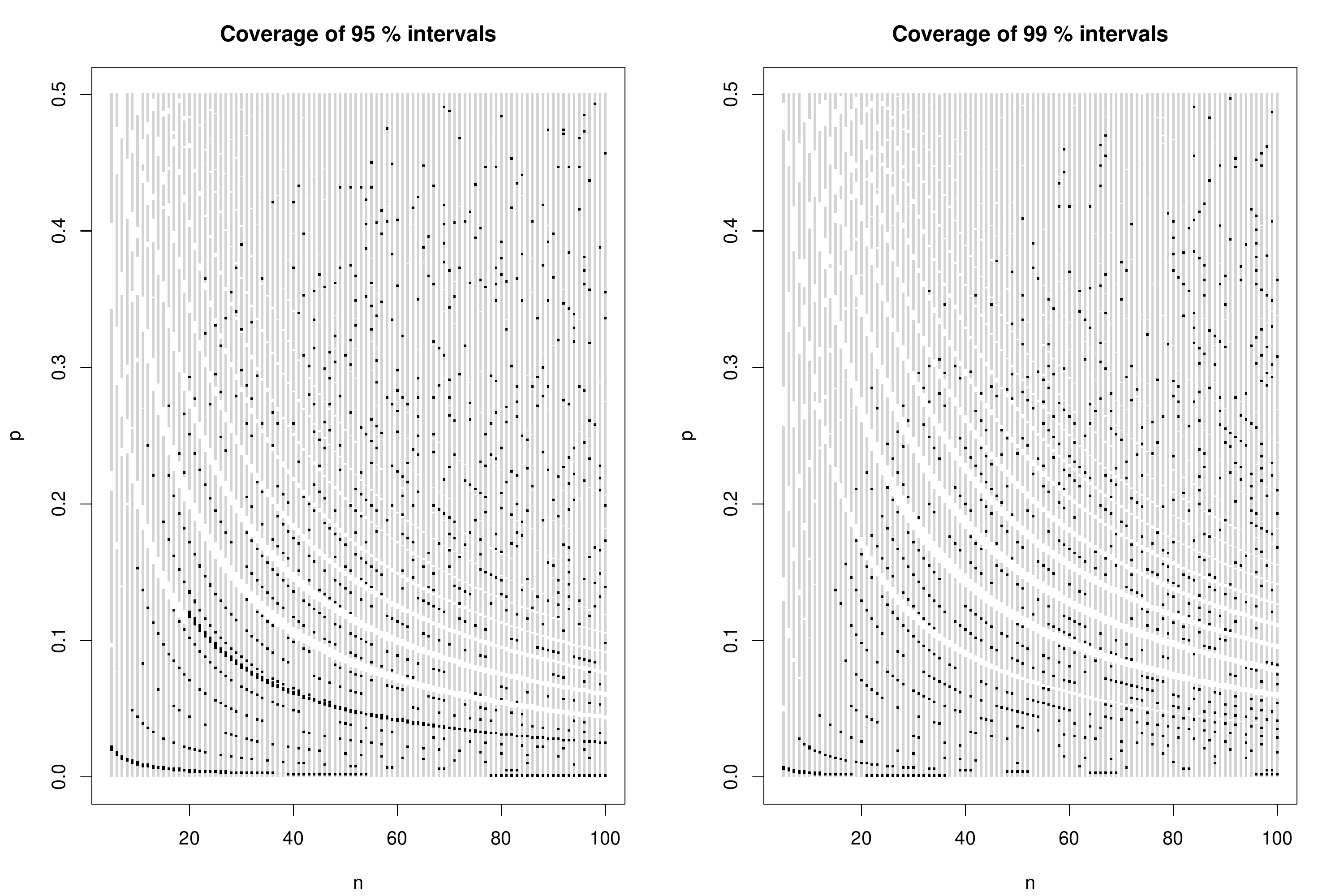}
   \caption{Comparison of the coverage of the Bayesian Jeffreys prior and posterior corrected Clopper-Pearson $Beta(0.5,0.5)$ intervals for different $n$ and $p$. In the black points the posterior corrected interval has greater coverage, in the white points the Bayesian Jeffreys prior interval has greater coverage and in the grey points the intervals have equal coverage (when rounded to 3 decimal places).}\label{figExp6}
\end{center}
\end{figure}

\pagebreak
\pagestyle{plain}

\end{document}